\documentclass{new_tlp}
\usepackage[utf8]{inputenc}
\usepackage[T1]{fontenc}
\usepackage{mathptmx}\usepackage{xspace}\usepackage{amssymb}
\usepackage{amsmath}\usepackage{prooftree}
\usepackage[pdftex]{color}
\input xypic

\newcommand\ciut{\hspace{.3mm}}

\def\<{\langle\,}\def\>{\,\rangle}
\newcommand{\kto}{\,{\to}\,}
\def\To{\Rightarrow}\def\Ot{\Leftarrow}

\newcommand\lupa{\mbox{\it loop\/}}
\def\FV{\mbox{\rm FV\/}}
\def\pusty{\varnothing}

\def\Rarrow{\To\kern -.650em \To}
\newcommand\klaus{\mathrel{{:}-}}

\newcommand{\sms}{\normalfont\textsc{sms}\xspace}

\newcommand{\ground}[1]{\mathit{ground}(#1)\xspace}

\def\Z{{\mathcal Z}}
\def\P{{\mathcal P}}
\def\Is{{\mathcal I}}\def\Mm{{\mathfrak M}}
\newcommand\erA{{\mathrm{Ans}}}\newcommand\erD{{\mathrm{Dsc}}}
\newcommand\erS{{\mathrm{Sbg}}}\newcommand\erH{{\mathrm{Hd}}}
\newcommand\erG{{\mathrm{Gol}}}\newcommand\erE{{\mathrm{Env}}}
\newcommand\erQ{{\mathrm{Qst}}}
\newcommand\rA{{\mathrm A}}\newcommand\rB{{\mathrm B}}

\newcommand\rF{{\mathrm F}}
\newcommand\rOm{{\mathrm\Omega}}
\newcommand\rK{{\mathrm K}}
\newcommand\rT{{\mathrm T}}
\newcommand\rP{{\mathrm P}}\newcommand\rQ{{\mathrm Q}}
\newcommand\rR{{\mathrm R}}\newcommand\rS{{\mathrm S}}
\newcommand\rY{{\mathrm Y}}
\newtheorem{theorem}{Theorem}
\newtheorem{corollary}[theorem]{Corollary}

\newtheorem{proposition}[theorem]{Proposition}

\def\wrt{with respect to\xspace}
\def\iff{if and only if\xspace}

  \title[First-order ASP]
        {First-order answer set programming as constructive proof search}

  \author[A. Schubert and P. Urzyczyn]
         {Aleksy Schubert and Paweł Urzyczyn\\
         University of Warsaw, Poland\\
         \email{[alx,urzy]@mimuw.edu.pl}}

\jdate{}
\pubyear{}
\pagerange{\pageref{firstpage}--\pageref{lastpage}}
\doi{}

\newtheorem{lemma}{Lemma}[section]

\begin{document}
\nocite{*}

\label{firstpage}

\maketitle

  \begin{abstract}
  We propose an interpretation of the first-order answer set
  programming (FOASP) in terms of intuitionistic proof theory.  
  It is obtained by two polynomial 
  translations between FOASP and  the bounded-arity fragment 
  of~the~$\Sigma_1$ level of the Mints hierarchy in~first-order 
  intuitionistic logic. It follows that $\Sigma_1$ formulas using
  predicates of fixed arity (in particular unary) is of the same 
  strength as FOASP. Our construction reveals a close similarity 
  between constructive provability and stable entailment,
  or equivalently, between the construction of an answer set and 
  an intuitionistic refutation. This paper is under consideration 
for publication in Theory and Practice of Logic Programming

  \end{abstract}


\tableofcontents
\section{Introduction}
The logic programming paradigm originates from the fundamental
idea that program execution can be realized as constructive 
proof search done by means of SLD resolution. A solution to 
a computational problem is obtained
in proof construction as a witness of the statement being proven.
While very natural and effective, this idea is 
hardly applied any further than to slight variations of Horn clauses.
(Works exploring the constructive approach, like~\cite{fukom17}, are not
very common.) 
Consider for example first-order ASP which is a form of adding
negation to Datalog. The standard way in which ASP is implemented
consists in 
encoding a~program into a satisfiability problem and then
applying a SAT-solver~\cite{Brewkaetal}. This clearly goes away 
from the proof search oriented motivation of logic programming.

In this paper we attempt to argue that this semantic approach 
to ASP is not necessary. We demonstrate a proof-theoretical 
interpretation of answer set programming by showing a~mutual
translation between first order ASP and the bounded arity $\Sigma_1$
fragment of intuitionistic predicate logic, as defined in~\cite{suz16}.
This extends our previous work~\cite{su-asp} where we have
defined a~similar equivalence between propositional answer set programming and
a~fragment of intuitionistic propositional logic. 
The extension required a~new approach since the techniques used 
in~\cite{su-asp} turned out to be not adequate \relax
in the first-order case. Of course, a first-order program can be
propositionalized by grounding, and then one could repeat the 
propositional construction to obtain a (quantifier-free) formula provable
iff the given entailment holds. However, the resulting formula would be no
longer polynomial in the size of the original program, as the
grounding is exponential in size. 

In addition, avoiding the exponential explosion is far not enough. 
The straightforward generalization of the propositional approach 
would yield a formula that does not fit our  {\it co-}{\sc NExptime} 
fragment of $\Sigma_1$, because targets of
subformulas would be of arbitrary arity. Instead 
we built our formula using "easy" axioms with nullary targets.
While the overall scheme ``unsound versus inconsistent'' 
remains
the same, the construction must be different. This is because
we can no longer use proof goals to control proof construction,
(in particular to discover loops in the unsound case~B) 
because there is too few of them. Instead we  must
use (significantly more non-deterministic) approach, 
continuation-passing in spirit.

Strictly speaking, the direct correspondence between proofs and
programs has a~``contravariant'' flavour in the case of ASP (also
known as stable model semantics). Indeed, a~stable model of a~program
does not represent a~proof of a~formula but rather a~{\it
  refutation\/}, i.e., a~certificate that a~proof does not exist. Put
it differently, provability in a~fragment of first-order
intuitionistic logic is equivalent to stable {\it entailment\/}, i.e.,
the ASP consequence relation.

The construction consists of two polynomial time translations.
First, for a~given first-order logic program $\P$ with negations, 
and a~given atom $\Omega$ we define a~formula~$\varphi$ so that 
$\P$ entails $\rOm$ under stable model semantics \iff $\varphi$
is an intuitionistic theorem (Proposition~\ref{111yyy}).

The formula~$\varphi$ is a~first-order formula~using only unary
predicate symbols  and only $\forall$ and~$\to$ as the logical
connectives. Since all quantifiers occur in $\varphi$ at negative
positions, our formula belongs to the $\Sigma_1$ fragment of Mints
hierarchy~\cite{suz16}.  In general, provability of
$\Sigma_1$-formulas is {\sc Nexpspace\/}-complete.  However, for any
fixed bound on the arity of predicate symbols, in particular for our
unary object language. 
 the provability problem turns out to be 
\mbox{{\it  co}-{\sc Nexptime}} complete.

Actually, the formula $\varphi$ constructed in Section~\ref{tamtam} 
is not yet unary.
It is an {\it easy\/} formula in the sense of~\cite{suz16}, because
all its non-atomic subformulas have nullary targets.
Provability of such formulas can however be reduced to provability 
with unary predicates. 

Our second construction reduces the question of refutability
(non-provability) of 
a~formula~$\varphi$ to the non-existence of a~stable model of a~logic
program with negations. Yet differently, given a~formula~$\varphi$, 
we define $\P$ so that
$\varphi$ is not an intuitionistic theorem \iff $\P$ has a~stable model
(Proposition~\ref{222xxx}).
The form of the latter equivalence is intended to stress that the stable 
models of $\P$ correspond to intuitionistic refutations.  The construction
works for $\Sigma_1$ formulas using predicate symbols of any bounded 
arity, not necessarily unary.

The two translations together add up to 
a polynomial equivalence of 
two problems known to be {\sc Nexptime}-complete: the existence 
of stable models and the provability for bounded-arity 
$\Sigma_1$ formulas.

As such, it does not constitute a formally new theorem
concerning complexity of the formalisms, and it was not 
meant to do so. The purpose and contribution of this paper is to
point out a natural constructive counterpart of the first-order
answer set programming, and by derivation of the correspondence to
expose the constructive nature of ASP.
In this way we complement the purely propositional accounts
\cite{Pearce06,su-asp} of the links between logic programming and
intuitionistic logic. In addition the current account is in
terms of proof theory and its details are closely related to the original
procedure of answer derivation in logic programming.

\section{Answer set programming and logic}\label{aspalogic}

First order answer set programming is an extension of Datalog
by negation understood as ``fixpoint''~\cite{kolpap}. The short presentation 
we give here is essentially following~\cite{Brewkaetal}.

\paragraph{Basic definitions:} A {\it positive atom\/}, or simply 
an~{\it atom\/}, is an atomic formula~$\rP(\vec x\ciut)$, where~$\vec x$ 
are individual variables or constants. A {\it negative atom\/}
has the form~$\neg\rP(\vec x)$, where $\rP(\vec x\ciut)$ is a~positive atom.
A {\it ground\/} atom is an atom where only constants can occur 
(and no variables).

A~{\it clause\/} is an expression of the form 
${\bf a}\klaus {\bf a}_1,\ldots,{\bf a}_n$,
where ${\bf a}$ is a~positive atom and ${\bf a}_1,\ldots,{\bf a}_n$ 
are positive or negative atoms. A~{\it ground\/} clause is one 
without variables. 
A {\it program\/} is a~finite set of clauses~$\P$, together with an
associated nonempty finite set of constants $D_\P$, the {\it domain
of\/}~$\P$, which is assumed to include all constants occurring in~$\P$. 

We write $B(\P)$ for the set of all positive ground atoms $\rP(\vec c\ciut)$,
where $\rP$ is a~predicate in the language of~$\P$ and $\vec c$ is a~vector
of constants in~$D_\P$ of appropriate length. 
By $\ground{\P}$ we denote the set of all ground clauses 
obtained from clauses of~$\P$ as substitution instances using 
constants in~$D_\P$. Note that $\ground{\P}$ can be seen
as a propositional program with negations using ground atoms in~$B(\P)$ 
as propositional literals. The following definitions are therefore immediate 
adaptations from the propositional case. 

A {\it model\/} of the language of~$\P$ is a set~$\Mm$ 
of positive ground atoms. 
One can think of $\Mm$ as of a~Herbrand structure
of domain~$D_\P$, satisfying the ground atoms in $\Mm$ 
and no other atom, or as of a~Boolean valuation of ground atoms. 
Given a program $\P$ and a model $\Mm$, we transform~$\ground{\P}$ 
into a program~$\P^\Mm$ without negations. For every ground atom~${\bf a}$:
\begin{itemize}
\item If ${\bf a}\not\in \Mm$ then we delete $\neg {\bf a}$ from the rhs of all 
clauses of~$\ground{\P}$;
\item If ${\bf a}\in \Mm$ then we delete all clauses of~$\ground{\P}$ 
with $\neg {\bf a}$ at the rhs.
\end{itemize}
The {\it interpretation\/} of $\P$ under $\Mm$, denoted $I(\P,\Mm)$, 
 is the least fixed point of the operator:
\begin{equation*}
  \label{eq:fixed-point}
  F(\Is)=\Is\cup\{{\bf a}\ |\ \mbox{there is 
a clause ${\bf a}\klaus {\bf a}_1,\ldots,{\bf a}_n$ in $\P^\Mm$
    such that all ${\bf a}_i$ are in\ }\Is\}.
\end{equation*}
A model $\Mm$ of~$\P$ is {\it stable\/} (or it is an {\it answer set\/}
for~$\P$) \iff \mbox{$\Mm=I(\P,\Mm)$}. We 
say that $\P$ {\it entails\/} an 
atom ${\bf a}$
{\it under SMS\/}, written \mbox{$\P\models_{\sms}{\bf a}$}, 
when every stable model
of~$\P$ satisfies~${\bf a}$. It is known~\cite{dantsin,kolpap} 
that the existence
of~a~stable model and the entailment under SMS are, respectively, 
$\mbox{\sc NExptime}$ and {\it co-\/}$\mbox{\sc NExptime}$ complete
problems.

\begin{lemma}\label{pljkuiuh}
Let $\rOm$ be a nullary predicate letter not occurring in a~program~$\P$. 
Then $\P\models_{\sms}\rOm$ \iff $\P$ has no stable model.
\end{lemma}

\paragraph{Intuitionistic logic:} We consider formulas of minimal predicate
logic with $\to$ and $\forall$ as the only connectives. Our formulas 
use neither function symbols nor equality. In particular, the only individual
terms are {\it object variables\/}, written in lower case, e.g.,~$x,y,\dots$
A formula is therefore either an {\it atom\/} 
$\rP(x_1,\ldots,x_n)$, where $n\geq 0$,
or an implication $\varphi\to\psi$, or it has the form $\forall x\,\varphi$.
\relax 
The convention is that $\to$ associates to the right, 
i.e.,~$\varphi\to\psi\to\vartheta$ stands for~$\varphi\to(\psi\to\vartheta)$.

Quantifier scopes are marked by parentheses or dots; for example \relax
$\forall x.\,\rP(x)\to \rQ(x)$ is the same as $\forall x(\rP(x)\to \rQ(x))$.
The set of free variables in a~formula $\varphi$ is defined as usual and
denoted~$FV(\varphi)$. The notation $\varphi[x:=y]$ stands for 
a~capture-avoiding substitution of $y$ for all free occurrences
of $x$ in $\varphi$. We also use the notion of simultaneous substitution
written $\varphi[\vec x:=\vec y]$.

A clause of the form ${\bf a} \klaus {\bf a}_1,\ldots, {\bf a}_n$ 
(without negations) is identified with the universal formula 
\mbox{$\forall \vec x({\bf a}_1\to\cdots\to {\bf a}_n\to {\bf a})$},
where $\vec x$ are all individual variables occurring in the clause. 
In~particular a~ground clause is an implicational closed formula, 
and a ground program can be seen as a set of such formulas.

Let $\overline \P$ be obtained from~$\ground{\P}$ 
by replacing all occurrences 
of negative atoms $\neg\rP(\vec c\ciut)$ by~$\overline\rP(\vec c\ciut)$,
where $\overline\rP$ is a new predicate 
symbol. 
Now let $\overline\Mm=\{\overline{\bf b}\ |\ 
{\bf b}\in B(\P)-\Mm\}$.
The following easy lemma (where $\vdash$ is intuitionistic
provability) is an immediate consequence of 
Lemma~4 in~\cite{su-asp}.

\begin{lemma}\label{dsewawqd}
$I(\P,\Mm)=\{{\bf a}\in B(\P)\ |\ \overline\P\cup\overline\Mm\vdash {\bf a}\}$.
\end{lemma}

\paragraph{Lambda-terms:}

We use lambda-terms for proof notation as e.g.~in~\cite{suw2013-www}.
In addition to object variables, used in formulas, we also have 
{\em proof variables\/} occurring in proofs. \relax
By convention, proof variables are written in upper case, 
e.g.,~$X,Y,\dots$ 
A finite set of declarations $(X:\varphi)$, where $X$ is a~proof 
variable and~$\varphi$ is a~formula, is called an {\it environment\/}
provided there is at most one declaration for any variable $X$.
A {\it proof term\/} (or simply ``term'') is one of the following:  
\begin{itemize} 
\item a proof variable,
\item an abstraction $\lambda X\,{:}\,\varphi.\,M$, where $\varphi$ is 
a~formula and $M$ is a~proof term, 
\item an abstraction $\lambda x\,M$, where $M$ is a proof term, 
\item an application $MN$, where $M$, $N$ are proof terms,
\item an application $Mx$, where $M$ is a proof 
term and $x$ is an object variable.
\end{itemize}  
That is, we have two kinds of lambda-abstraction: the 
proof abstraction \mbox{$\lambda X\ciut{:}\,\varphi.\,M$} and
the object abstraction $\lambda x\,M$.
There are also two forms of application: the proof application $MN$, where $N$
is a proof term, and the object application~$My$, where~$y$ is an object
variable. Terms (and also formulas) are taken up to alpha-conversion, i.e.,
the choice of bound variables is irrelevant.

The type-assignment rules in Figure~\ref{prufassirules} infer 
judgments of the form $\Gamma \vdash M: \varphi$, where~$\Gamma$ 
is an environment, $M$ is a~term, and $\varphi$ is a formula.
In rule $({\forall}{I})$ we require $x \not \in FV(\psi)$ for 
every $\psi\in\Gamma$.

The formalism is used liberally. For instance, we may say 
that ``a term $M$ has type $\alpha$'' leaving the environment implicit. 
We also may identify proof environments with the appropriate sets of formulas. 

The notion of a term in {\it long normal form\/} (lnf) 
is defined according to its type in a given environment.
\relax
\begin{itemize}
\item If $N$ is an lnf of type $\varphi$ then $\lambda x\, N$ is an lnf of
  type $\forall x\,\varphi$.\relax
\item If $N$ is an lnf of type $\psi$ then 
$\lambda X\,{:}\,\varphi.\, N$  is an lnf 
of type $\varphi \to \psi$.
\item If $N_1,\ldots, N_n$ are lnf or object variables, and $XN_1\ldots
  N_n$ is of an atom type, then $XN_1\ldots N_n$ is an lnf.
\end{itemize}
The following lemma makes it possible to restrict proof search 
to long normal forms. 

\begin{lemma}[Schubert et~al\mbox{. }{2015}] 
\label{lemma:lnf}
If $\varphi$ is intuitionistically derivable from $\Gamma$ 
then $\Gamma \vdash N:\varphi$, for some lnf~$N$.
\end{lemma}

\begin{figure}[h]
$$
\Gamma, X: \varphi ~\vdash X:\varphi \quad({A}x)
$$
$$\begin{prooftree}
\Gamma, X: \varphi ~\vdash M:\psi
\justifies
\Gamma ~\vdash \lambda X\ciut{:}\ciut\varphi.\ciut M\ :\ \varphi \to \psi
\thickness=0.08em
\using
	({\to}{I})
\end{prooftree}
\qquad
\begin{prooftree}
\Gamma ~\vdash M : \varphi \to \psi \quad
\Gamma ~\vdash N : \varphi
\justifies
\Gamma ~\vdash MN : \psi
\thickness=0.08em
\using
	({\to}{E})
\end{prooftree}$$
$$\begin{prooftree}
\Gamma ~\vdash M : \varphi
\justifies
\Gamma ~\vdash \lambda x\,M : \forall x \varphi
\thickness=0.08em
\using
	({\forall}{I})
\end{prooftree}
\qquad
\begin{prooftree}
\Gamma ~\vdash M : \forall x \varphi
\justifies
\Gamma ~\vdash My : \varphi[x:=y]
\thickness=0.08em
\using
	({\forall}{E})
\end{prooftree}$$~
\caption{Proof assignment rules\label{prufassirules}}
\end{figure}

\noindent
We are concerned with the $\Sigma_1$ level of the Mints 
hierarchy, defined by: 
\begin{itemize}\parskip=0pt
\item $\Sigma_{1} ::= {\bf a}\ |\ \Pi_{1}\to \Sigma_{1}$; 
\item $\Pi_{1} ::= {\bf a}\ |\ \Sigma_{1}\to\Pi_{1}\ |\ 
\forall x\,\Pi_{1}\,$,\hfill
\end{itemize}
where ${\bf a}$ stands for an atomic formula.

Observe that every $\Sigma_1$-formula takes the form 
$\varphi=\tau_1\to\tau_2\to\cdots\to\tau_q\to{\bf c}$, where all $\tau_i$ are
$\Pi_1$-formulas and ${\bf c}$ is an atom, called the 
{\it target\/} of~$\varphi$. The following follows from 
Lemma~\ref{lemma:lnf}.

\begin{lemma}\label{genera1}
If $\varphi=\tau_1\to\tau_2\to\cdots\to\tau_q\to{\bf c}$ then 
$\Gamma \vdash \varphi$ is derivable
\iff so is $\Gamma,\tau_1,\ldots,\tau_q\vdash {\bf c}$.
\end{lemma}
\begin{proof}
A long normal proof of $\varphi$ (i.e., a lnf of type~$\varphi$)
must have the form $\lambda X_1\,{:}\,\tau_1\dots\lambda X_q\,{:}\,\tau_q.\,N$,
where the term $N$ has type~${\bf c}$.
\end{proof}

Thus to determine if a given $\Sigma_1$-formula is provable, one 
has to consider judgments $\Gamma \vdash {\bf c}$, where all 
members of~$\Gamma$ are in $\Pi_1$, i.e., of the form
$\forall\vec y_1(\sigma_1\to \forall\vec y_2(\sigma_2\to 
\cdots\to\forall\vec y_k(\sigma_k\to{\bf b})\ldots))$, where 
$\sigma_i\in \Sigma_1$, and $\,{\bf b}$ is an atom. The atom $\,{\bf b}$
is called the {\it target\/} of~$\psi$.

From Lemmas~\ref{lemma:lnf} and~\ref{genera1} we now obtain 
the following ``generation lemma'':

\begin{lemma}\label{genera22}
Let ${\bf c}$ be an atom, and let $\Gamma$ consists of $\Pi_1$ formulas. 
Then $\Gamma \vdash {\bf c}$
is derivable iff there is
a formula in~$\Gamma$ of the form
$\psi=\forall\vec y_1(\sigma_1\to \forall\vec y_2(\sigma_2\to 
\cdots\to\forall\vec y_k(\sigma_k\to{\bf b})\ldots))$,
and variables $\vec x=\vec x_1\ldots,\vec x_k$ 
such that ${\bf b}[\vec y:=\vec x]={\bf c}$, where 
$\vec y=\vec y_1\ldots\vec y_k$, and, for~$i=1\,\ldots,k$,
if \mbox{$\sigma_i=\tau_1^i\to\tau_2^i\to\cdots\to\tau_{q_i}^i\to{\bf a}_i$}
then all judgments 
$\Gamma,\tau_1^i[\vec y:=\vec x],\ldots,\tau_{q_i}^i[\vec y:=\vec x]\vdash 
{\bf a}_i[\vec y:=\vec x]$
are derivable.
\end{lemma}

Lemma~\ref{genera22} defines a Ben-Yelles style proof-search algorithm, which
can be seen as a behaviour of an alternating automaton. One interprets 
 $\Gamma \vdash {\bf c}$ as a configuration of the machine, where 
{\bf c} is the internal state and $\Gamma$ is (write-only) memory or database. 
A computation  step consists of a nondeterministic choice of 
an assumption~$\psi$ from the database, together with an appropriate 
substitution, and of a universal split into as many computation branches as 
there are premises in~$\psi$. Each of these commences with the proof goal
(internal state) given by a target of one of the premises in~$\psi$.

\section{SMS entailment as provability}\label{tamtam}
Given a program $\P$ with negations, and a ground atom $\rOm$,
we define in this section a formula $\varphi$ such that
$\P$ entails $\rOm$ under stable model semantics \iff $\varphi$ is
intuitionistically provable. We begin with some preparatory 
considerations regarding the (now fixed) program~$\P$.

We assume that $D_\P=\{c_1,\ldots,c_m\}$, and 
every clause $K$ of~$\P$ is written as

\hfil $\rR(\vec u)\klaus \rP_1(\vec v_1),\ldots,\rP_r(\vec v_r),
\neg\rS_1(\vec w_1),\ldots,\neg\rS_s(\vec w_s)$,

\noindent
where $\vec u,\vec v_1,\ldots,\vec v_r,\vec w_1,\ldots,\vec w_s$ are 
some vectors of variables and constants, possibly with repetitions.
The {\it arity\/} of clause~$K$ is the arity of~$\rR$ 
(the length
of~$\vec u$). 

Let $\vec x$ be the sequence of all variables~in~$\vec u$ 
taken without repetitions. Then let
$\vec y$ stand for the sequence of all variables
 occurring in 
$\vec v_1,\ldots,\vec v_r,\vec w_1,\ldots,\vec w_s$,
but not occurring in~$\vec x$, 
listed without repetitions in some fixed order. We may now 
denote the clause~$K$ by $K(\vec x,\vec y)$. The vector $\vec y$
consists of variables that only occur at the rhs of~$K$.

Recall that program $\overline \P$ is obtained from~$\ground{\P}$ 
by replacing all occurrences 
of negative atoms $\neg\rP(\vec c\ciut)$ by~$\overline\rP(\vec c\ciut)$.
The program $\overline \P$ thus consists of clauses $K(\vec a,\vec b\ciut)$ 
of the form 

\hfil $\rR(\vec e\ciut)\klaus \rP_1(\vec c_1),\ldots,\rP_r(\vec c_r),
\overline\rS_1(\vec d_1),\ldots,\overline\rS_s(\vec d_s)$, 

\noindent
obtained from $K(\vec x,\vec y)$ by substituting elements 
$\vec a,\vec b$ of $D_\P$ for $\vec x,\vec y$. 
Vectors  $\vec e$,~$\vec c_i$, and $\vec d_i$ are appropriately selected from 
$\vec a,\vec b$.  More precisely, if $S$ denotes the 
substitution $[\vec x,\vec y := \vec a,\vec b]$, then $\vec e=\vec u[S]$,
$\vec c_i=\vec v_i[S]$ and $\vec d_i=\vec w_i[S]$.

The vocabulary of $\varphi$ consists of~$D_\P$ 
and the following predicate symbols:

\begin{itemize}
\item For every predicate symbol $\rP$ in $\P$, there are four 
symbols $\rP$, $\overline\rP$, $\rP!$, and $\rP?$, each of the same
arity as~$\rP$.

\item For every pair $\rP$, $\rQ$ of predicate symbols, of arity $m,n$
respectively, there 
is a~predicate $\rP\rQ$ of arity $m+n$.

\item For every clause $K(\vec x,\vec y)$ of arity $l$ with 
$|\vec y|=m$, and for every $i=0,\ldots,m$ there is a~predicate 
$\rK^i$ of arity $l+i$, and a nullary predicate $\overline\rK^i$.

\item In addition we have six nullary predicates 
$\lupa$, $\rOm$, $\rA$, $\rB$, $\circ$, and $\bullet$.
\end{itemize}
The arity of predicates occurring in $\varphi$ depends on the 
arity of atoms in~$\P$. However, all targets of the implication 
subformulas in~$\varphi$ are nullary. As such, $\varphi$ is an ``easy
formula'' and by~\cite{suz16} it is translatable in polynomial time
to a~formula~$\varphi'$ with unary predicates so that $\varphi'$ is provable 
\iff so is~$\varphi$.

\begin{figure}[h!]
\begin{enumerate} \normalsize
\item\label{aksj11} 
$\forall\vec z((\rR(\vec z\ciut)\to{\lupa})
\to(\overline\rR(\vec z\ciut)\to{\lupa})
\to {\lupa})$, \quad for every predicate symbol $\rR$ of~$\P$.

\item\label{aksj22} $\rOm\to\lupa$, ~~$\rA\to \lupa$, ~~and $\rB\to \lupa$.

\item\label{aksj33} $\forall\vec z.\, 
\overline\rR(\vec z\ciut)\to (\rR!(\vec z\ciut)\to\bullet)\to\rA$,
\quad for every predicate symbol $\rR$ of~$\P$.

\item\label{aksj00} 
$\forall\vec z.\, \rR(\vec z\ciut)\to (\rR?(\vec z\ciut)\to\circ)\to \rB$,
\quad for every predicate symbol $\rR$ of~$\P$.

\item\label{aksj44} $\forall\vec x\vec y.\,\rR!(\vec u)\to 
(\rP_1!(\vec v_1)\to\bullet)\to\cdots\to(\rP_r!(\vec v_r)\to\bullet)\to
\overline\rS_1(\vec w_1)\to\cdots\to\overline\rS_t(\vec w_s)\to\bullet$, \\
for every 
$K(\vec x,\vec y)$ of the form 
\mbox{$\rR(\vec u)\klaus \rP_1(\vec v_1),\ldots,\rP_r(\vec v_r),
\neg\rS_1(\vec w_1),\ldots,\neg\rS_t(\vec w_t)$}.

\item\label{aksj55} $\forall\vec z.\,\rR?(\vec z\ciut)\to%
  (\rK^0_1(\vec z\ciut)\to\overline\rK^0_1)%
  \to\cdots\to 
  (\rK^0_l(\vec z\ciut)\to\overline\rK^0_l)
  \to\circ$,\quad for every predicate symbol $\rR$ such that
$K_1,\ldots,K_l$ are all clauses with target~$\rR$.

\item\label{aksjaaa} $\forall\vec z_1\vec z_2.\,\rK^0(\vec z_1,d,\vec z_2)\to
\overline\rK^0$,  \quad for every constant $d\neq c$, and every 
clause $K$ with target~$\rR(\vec u)$, where $c$ occurs in~$\vec u$ as the
$|\vec z_1|+1$-st argument.

\item\label{aksjbbb}$\forall\vec z_1\vec z_2\vec z_3.\,
\rK^0(\vec z_1,c,\vec z_2,d,\vec z_3)\to\overline\rK^0$,  \quad for every 
pair $c,d$ of different constants, and every 
clause~$K$ with target~$\rR(\vec u)$, where the 
$|\vec z_1|+1$-st and the $|\vec z_1|+1+|\vec z_2|+1$-st 
arguments in $\vec u$ are the same variable.

\item\label{aksj66} $\forall\vec z\vec\nu.\,
\rK^i(\vec z,\vec\nu)\to
(\rK^{i+1}(\vec z,\vec\nu,c_1)\to \overline\rK^{i+1})\to\cdots\to
(\rK^{i+1}(\vec z,\vec\nu,c_m)\to \overline\rK^{i+1})\to \overline\rK^i$, \\
for every clause $K(\vec x,\vec y)$, and every prefix $\vec\nu$ of~$\vec y$
such that $|\vec\nu|=i<|\vec y|$.\\
(Recall that $c_1,\ldots,c_m$ are all the constants occurring in~$\P$.)

\item\label{aksj77} $\forall\vec x\vec y.\,
\rK^m(\vec u,\vec y)\to
(\rP?(\vec v)\to  \rR\rP(\vec u,\vec v) \to \circ)\to
    \overline\rK^m$,  \quad
for every clause $K(\vec x,\vec y)$ with target~$\rR(\vec u)$ 
such that $\rP(\vec v)$ occurs at the rhs of~$K(\vec x,\vec y)$.

\item\label{aksj88} $\forall\vec x\vec y.\,
\rK^m(\vec u,\vec y)\to S(\vec w)\to \overline\rK^m$, \quad
for every clause $K(\vec x,\vec y)$ with target~$\rR(\vec u)$ 
such that $\neg\rS(\vec w)$ occurs at the rhs of~$K(\vec x,\vec y)$.

\item\label{aksj99} 
$\forall\vec z\vec y\vec w (\rR\rP(\vec z,\vec y)\to \rP\rQ(\vec y,\vec w)\to 
  (\rR\rQ(\vec z,\vec w)\to\circ)
  \to\circ)$, \quad
for every predicate symbols $\rR,\rP,\rQ$ of $\P$ of respective arity 
$|\vec z|$, $|\vec y|$, $|\vec w|$.

\item\label{aksjAA} $\forall\vec z (\rP\rP(\vec z,\vec z\ciut)\to \circ)$,
 \quad for every  predicate symbol $\rP$ of $\P$.
\end{enumerate}
\caption{Axioms of the formula~$\varphi$\label{figaxioms}.
Vectors $\vec z,\vec z_i$ always consist of different variables.}
\end{figure}

\noindent The formula 
$\varphi$ has the form $\psi_1\to\cdots\to\psi_d\to\lupa$, where 
$\lupa$ is a nullary predicate symbol, and $\psi_1,\ldots,\psi_d$ are
closed formulas, called {\it axioms\/}, listed in Figure~\ref{figaxioms}.
We now explain how these axioms can be used to prove~$\varphi$.

First we have the initial axioms~(\ref{aksj11}), 
for every predicate letter~$\rR$.
These axioms can be applied towards proving the main goal~$\lupa$. 
Every time an axiom of type~(\ref{aksj11}) is used, the proof universally 
splits into two branches and at each of them a ground atom either of the 
form $\rR(\vec c\ciut)$ 
or of the form $\overline\rR(\vec c\ciut)$ is added to the proof environment. 
At the end of this phase we are facing the task of proving $\lupa$ under
any possible choice of positive or overlined ground facts. More precisely,
every branch of the proof leads to a judgment $\Gamma\vdash\lupa$, where
$\Gamma$ contains a~selection of ground atoms in addition to the axioms 
of~$\varphi$. 
Ideally we could make exactly as many choices as is needed to fully determine
a~different model~$\Mm$ at every branch. This happens when 
axiom~(\ref{aksj11}) is 
used at every branch exactly once for every predicate letter $\rR$ and 
every substitution $\vec z := \vec c$. Then the set~$\Gamma$ contains
either $\rR(\vec c\ciut)$ or $\overline\rP(\vec c\ciut)$, for every~$\rR$
and every~$\vec c$.  
 We could achieve this in the propositional case~\cite{su-asp} by using 
all propositional atoms as proof goals one after one in a deterministic
fashion. As we now work with nullary goals, this is impossible, and we must 
accept nondeterminism. So in general on any given
branch we may either have a ``partial model'' (when neither $\rR(\vec c\ciut)$ 
nor~$\overline\rR(\vec c\ciut)$ was chosen) or an ``inconsistency'' (when 
both are chosen). In the latter case a certain branch of the proof 
leads to a judgment of the form
$\Gamma,\rR(\vec c\ciut),\overline\rR(\vec c\ciut)\vdash\lupa$. 
Observe though that such an inconsistency is not 
as dangerous as it may appear. It only happens when the same 
axiom~(\ref{aksj11})
was used twice for the same ground substitution $\vec z:=\vec c$. 
Then in addition to the ``inconsistent'' branch we have 
two other branches where both $\rR(\vec c\ciut)$ 
and $\overline\rR(\vec c\ciut)$ were selected in a~consistent manner.
A~successful proof must therefore always handle all the consistent cases, 
i.e., apply to all models. 

The goal $\lupa$ should be provable when $\P$ entails $\rOm$. The 
three axioms~(\ref{aksj22}):

\hfil $\rOm\to\lupa$,\hfil $\rA\to \lupa$,\hfil $\rB\to \lupa$,

\noindent
where $\rA,\rB$ are fixed nullary predicate letters, 
correspond to the three ways in which the entailment $\P\models_{\sms} \rOm$
may hold in a stable model: either $\rOm$ holds, or the model is 
unstable because $\overline\P\cup\overline\Mm$ proves too much 
($\P$ is unsound for~$\Mm$), 
or the model is unstable because $\overline\P\cup\overline\Mm$ 
does not prove what 
is needed ($\P$ is incomplete for~$\Mm$). In the first case, a proof 
of $\varphi$ should be completed with help of the axiom $\rOm\to\lupa$,
as the atom~$\rOm$ (a member of $B(\P)$) is available as an assumption.
Otherwise the proof goal is set to $\rA$ or $\rB$, respectively. \bigbreak

\paragraph{The unsound case~$\rA:$ }~
For the unsound case we have a nullary predicate $\bullet$ and 
a~collection of axioms~(\ref{aksj33}) of the form 

\hfil  $\forall\vec z.\, 
\overline\rR(\vec z\ciut)\to (\rR!(\vec z\ciut)\to\bullet)\to\rA$.

\noindent
Only one of axioms (3) can be used at every branch of any proof and only
once. In our automata-theoretic interpretation, \relax
it passes control from $\rA$ to $\bullet$, adding $\rR!(\vec c\ciut)$ to
the memory, for some instance~$\vec c$ of~$\vec z$. 
The following proof should succeed if~$\rR(\vec c\ciut)\in I(\P,\Mm)$,
despite that $\overline\rR(\vec c\ciut)\in\overline\Mm$.

The intuitive meaning of the predicate $\rR!(\vec c\ciut)$
is: ``the atom  $\rR(\vec c\ciut)$ has been visited as a~proof goal
in a derivation from $\overline\P\cup\overline\Mm$\ciut'' 
(cf.~Lemma~\ref{dsewawqd}).

Now consider any clause of~$\P$ with target $\rR$, e.g., 
$\rR(\vec x)\klaus  \rP(\vec x),\rQ(\vec x),\neg \rS(\vec x)$,
where~$\vec x$ consists only of variables  and has no repetitions. (That 
is assumed for simplicity of the example.) 
With such a clause we associate an assumption of type~(\ref{aksj44}) which
in this case takes the form 

\hfil
 $\forall\vec x.\, \rR!(\vec x)\to (\rP!(\vec x\ciut)\to\bullet)\to 
(\rQ!(\vec x\ciut)\to\bullet)\to \overline \rS(\vec x\ciut)\to\bullet$,

\noindent and enforces 
a universal split to two new tasks: (i) prove $\bullet$ under the
assumption~$\rP!(\vec c\ciut)$, and (ii) prove $\bullet$ under the 
assumption~$\rQ!(\vec c\ciut)$, provided  $\rR(\vec c\ciut)$ has 
already been visited 
as a proof goal and 
$\overline \rS(\vec c\ciut)$ holds in the model. Axioms of type~(\ref{aksj44})
are now repeatedly applied in every branch of the proof  until $\bullet$ 
can be derived without creating a new proof task. This happens when 
an assumption of the form $\rR!(\vec d\ciut)$ corresponds to a clause
with no positive atom at the rhs. The next lemma makes this explicit.

For a  model $\Mm$, 
let $\Gamma_\Mm$ be the union of~$\Mm\cup\overline\Mm$
and the set of all axioms of~$\varphi$.\bigbreak

\begin{lemma}\label{kkff4433} Let $\Delta$ consist of assumptions of the form
$\rR!(\vec c\ciut)$. Then $\Gamma_\Mm,\Delta\vdash \bullet$ holds \iff
there is an atom $\rR!(\vec c\ciut)\in\Delta$ such that 
$\overline\P\cup\overline\Mm\vdash\rR(\vec c\ciut)$.
\end{lemma}

\begin{proof} $(\To)$ Induction wrt the size of a long normal proof
(cf.~Lemma~\ref{lemma:lnf}). 
To prove $\bullet$ one must use an instance of one of the 
axioms~(\ref{aksj44}), say of the form:\vspace{1mm}

\noindent \hfill
$\rR!(\vec a\ciut)\kto (\rP_1!(\vec c_1\ciut)\kto\bullet)\kto \cdots
\kto(\rP_r!(\vec c_r\ciut)\kto\bullet)\kto 
\overline\rS_1(\vec d_1\ciut)\kto\cdots\kto 
\overline\rS_t(\vec d_s\ciut)\kto\bullet$\hfill (*)\vspace{1mm}

\noindent Then the atoms
$\rR!(\vec a), 
\overline\rS_1(\vec d_1),\ldots,\overline\rS_s(\vec d_s)$  are provable 
from $\Gamma_\Mm,\Delta$, and this can only happen when 
$\rR!(\vec a)\in \Delta$
and $\overline\rS_1(\vec d_1),\ldots,\overline\rS_s(\vec d_s)\in\overline\Mm$, 
as there are no other assumptions with these targets.

In addition there are proofs of 
$\Gamma_\Mm,\Delta,\rP_i!(\vec c_i\ciut)\vdash\bullet$\,, for $i=1,\ldots,r$.
By the ind.~hypothesis, there are atoms
$\rR_i!(\vec e_i\ciut)\in \Delta,\rP_i!(\vec c_i\ciut)$ such that 
$\overline\P\cup\overline\Mm\vdash\rR_i(\vec e_i\ciut)$. 
Should any of these belong
to~$\Delta$, we are done. Otherwise $\rR_i(\vec e_i\ciut)=
\rP_i(\vec c_i\ciut)$, for all~$i$. (This covers the ``base case'' when~$r=0$
and the induction hypothesis is not used.)
Then we can derive $\rR(\vec a)$ from the 
appropriate clause in~$\overline\P\cup\overline\Mm$, which is
$\rR(\vec a)\klaus \rP_1(\vec c_1\ciut),\ldots,\rP_r(\vec c_r\ciut),
\overline\rS_1(\vec d_1\ciut),\ldots,\overline\rS_s(\vec d_s\ciut)$.

$(\Ot)$ Induction \wrt the size of a long normal proof.
The proof must use a clause 
$\rR(\vec a)\klaus \rP_1(\vec c_1\ciut),\ldots,\rP_r(\vec c_r\ciut),
\overline\rS_1(\vec d_1\ciut),\ldots,\overline\rS_s(\vec d_s\ciut)$
which corresponds to an axiom (\ref{aksj44}) instance of the form~(*).
Then $\overline\rS_1(\vec d_1\ciut),\ldots,
\overline\rS_s(\vec d_s\ciut)\in\overline\Mm$
and $\overline\P\cup\overline\Mm\vdash\rP_i(\vec c_i\ciut)$, 
for every $i=1,\ldots,r$.
By the induction hypothesis we have 
$\Gamma_\Mm,\Delta,\rP_i!(\vec c_i\ciut)\vdash \bullet$\,, 
for all~$i$. This enables us to apply (*) 
to obtain a proof of~$\bullet$ from~$\Gamma_\Mm,\Delta$.\relax
\end{proof}

The following is an easy consequence 
of Lemmas~\ref{dsewawqd} and~\ref{kkff4433}.

\begin{lemma}\label{jwq7654}
We have $\Gamma_\Mm\vdash\rA$ \iff there is $\rR(\vec c\ciut)\in I(\P,\Mm)$
such that $\overline\rR(\vec c\ciut)\in\overline\Mm$. That is, 
$\Gamma_\Mm\vdash\rA$ \iff $\P$ is unsound for~$\Mm$.
\end{lemma}
\begin{proof} 
~$(\To)$~~The only way to prove $\Gamma_\Mm\vdash\rA$ is by
applying one of the axioms~(\ref{aksj33}), for some $\rR(\vec c\ciut)$.
Then $\overline\rR(\vec c\ciut)\in\Gamma_\Mm$ 
and there must be a proof 
of \mbox{$\Gamma_\Mm,\rR!(\vec c\ciut)\vdash \bullet$}.
We apply Lemma~\ref{kkff4433} for the one element set 
$\Delta=\{\rR!(\vec c\ciut)\}$. \relax 

$(\Ot)$ Let $\P$ be unsound for~$\Mm$. There is an atom
$\rR(\vec c\ciut)\in I(\P,\Mm)$ such that 
$\overline\rR(\vec c\ciut)\in\overline\Mm$.
Then we have $\overline\P\cup\overline\Mm\vdash\rR(\vec c\ciut)$ by
Lemma~\ref{dsewawqd} whence $\Gamma_\Mm,\rR!(\vec c\ciut)\vdash\bullet$ 
by Lemma~\ref{kkff4433}. To prove~$\rA$ one uses an instance of 
axiom~(\ref{aksj33}).\relax
\end{proof}

\paragraph{The incomplete case~$\rB$:}~
The atom $\rB$ can be derived with the help of 
(\ref{aksj00}):

\hfil $\forall\vec z.\, \rR(\vec z\ciut)\to (\rR?(\vec z\ciut)\to\circ)\to\rB$.

\noindent
This axiom can be used when $\rR(\vec c\ciut)\in\Mm$, for some~$\vec c$,
and it yields the proof task $\Gamma_\Mm, \rR?(\vec c\ciut)\vdash\circ$. 
The meaning of the assumption $\rR?(\vec c\ciut)$ is ``we claim that 
$\rR(\vec c\ciut)$ cannot be derived from $\overline\P\cup\overline\Mm$''.
Now if $K_1,\dots,K_l$
are all clauses with target $\rR(\vec x)$ then we can use 
an instance of axiom~(\ref{aksj55}) 

\hfil  $\rR?(\vec c\ciut)\to%
  (\rK^0_1(\vec c\ciut)\to\overline\rK^0_1)%
  \to\cdots\to 
  (\rK^0_l(\vec c\ciut)\to\overline\rK^0_l)
  \to\circ.$

\noindent This results in $l$ parallel proof obligations 
$\Gamma_\Mm,\rR?(\vec c\ciut),\rK^0_i(\vec c\ciut)\vdash\overline\rK^0_i$, 
for all $i\leq l$. 
The intuitive understanding of the assumption $\rK^0_i(\vec c\ciut)$
is ``clause $K_i$ cannot be 
used to derive $\rR?(\vec c\ciut)$''. 

We need some care to handle the additional free variables
in~$K_1,\dots,K_l$, and this is why we need axioms of type~(\ref{aksj66}).
 For example, if~clause $K$ is one of $K_1,\dots,K_l$, and 
it has free variables $y_1,y_2,y_3$ occurring only at the rhs, then 
we need three additional assumptions to extend the vector~$\vec z$.
So we take the following steps (recall that 
$c_1,\ldots,c_m$ are all the constants):

\noindent$\forall\vec z.\,\rK^0(\vec z\ciut)\to 
(\rK^1(\vec z,c_1)\to \overline \rK^1)\to
\cdots\to(\rK^1(\vec z,c_m)\to \overline \rK^1)\to\overline\rK^0$

\noindent$\forall\vec z y_1.\, \rK^1(\vec z,y_1)\to 
(\rK^2(\vec z,y_1,c_1)\to \overline\rK^2)\to
\cdots\to(\rK^2(\vec z,y_1,c_m)\to \overline \rK^2)\to \overline \rK^1$

\noindent$\forall\vec z y_1 y_2.\, \rK^2(\vec z,y_1,y_2)\kto 
(\rK^3(\vec z,y_1,y_2,c_1)\kto \overline \rK^3)\kto
{\cdot}{\cdot}{\cdot}%
(\rK^3(\vec z,y_1,y_2,c_m)\kto 
\overline \rK^3)\kto 
\overline \rK^2$

Note that with~$m$ elements of the domain we could ``compress''
the three example formulas into one with~$m^3$ premises,
but we cannot do it in general: the number of added variables 
in a~clause can be proportional to the size of the program.
We have avoided that using $n$ formulas with $m$ premises each 
rather than one formula with $m^n$ premises.

The auxiliary formulas, targeted $\circ, \overline\rK^1, 
\overline\rK^2$,
generate $m^3$ independent processes, each with a~different 
``valuation'' of $y_1,y_2,y_3$. For every such valuation one of the 
assumptions targeted~$\overline \rK^3$ must be fired to guarantee 
that clause $K$
cannot be used towards a proof of $\rR(\vec x)$.

For example let $K$ has the form $\rR(\vec x)\klaus \rP(\vec w), 
\rQ(\vec\nu),\neg \rS(\vec v)$, where $\vec x$ are different variables.
Then all variables $\vec w,\vec\nu, \vec v$
are among $\vec x,\vec y$, and we can use axioms of type~(\ref{aksj77})
and~(\ref{aksj88}):

\hfil $  \begin{array}{l}
  \forall\vec x\vec y.\,\rK^3(\vec x,\vec y)\to
  (\rP?(\vec w)\to  \rR\rP(\vec x,\vec w) \to \circ)\to
    \overline\rK^3\\

  \forall\vec x\vec y.\,\rK^3(\vec x,\vec y)\to
  (\rQ?(\vec\nu)\to \rR\rQ(\vec x,\vec\nu)\to \circ)\to \overline \rK^3\\

  \forall\vec x\vec y.\,\rK^3(\vec x,\vec y)\to\; %
    \rS(\vec v)\to \overline \rK^3    
  \end{array}$

\noindent
Predicates denoted  $\rR\rP(\vec x,\vec w), \rR\rQ(\vec x,\vec\nu)$ 
``remember'' about the transition from the goal $\rR?(\vec x)$ to the goal 
$\rP?(\vec w)$, respectively $\rQ?(\vec\nu)$. The axiom~(\ref{aksj99})
ensures that this memory is transitive, and axiom~(\ref{aksjAA}) terminates
the proof construction when a loop in memory is revealed. 

An additional source of complication is that the lhs of a clause
can contain constants or repeated variables. Such a clause can therefore
be used only towards matching goals, and this is handled by 
axioms~(\ref{aksjaaa}) and~(\ref{aksjbbb}). For example, if clause~$K$
begins with $\rR(x,y,x)\klaus\cdots$, and an assumption $\rK^0(a,b,c)$ 
was introduced with $a\neq c$, 
then the proof can be immediately completed using
axiom~(\ref{aksjbbb}) of the form $\forall z.\, \rK^0(a,z,c)\to \overline\rK^0$.

We now show that a proof of $\rB$ is possible \iff our program~$\P$
is incomplete. We begin with the ``if'' part.

\paragraph{Refutations} Fix an $\Mm$. Let ${\bf a}\in B(\P)$. 
A~{\it refutation for\/}~${\bf a}$ is a possibly infinite 
tree labeled by members of~$B(\P)$ 
such that the root has label~${\bf a}$ and 
  as many immediate subtrees as there are clauses in 
$\ground{\P}$ with target~${\bf a}$. 
For every such clause, the corresponding subtree 
\begin{itemize} 
\item either is a refutation for ${\bf b}$,
for some positive atom ${\bf b}$ occurring at the rhs of the clause, 
\item or consists of a single node labeled $\overline{\bf b}$, 
where ${\bf b}\in \Mm$ and $\neg{\bf b}$ occurs at the rhs of the clause.
\end{itemize}

\begin{lemma}\label{uud352}
If $\rR(\vec c\ciut)\not\in I(\P,\Mm)$ then there exists a refutation 
for $\rR(\vec c\ciut)$.
\end{lemma}
\begin{proof}
Suppose $\rR(\vec c\ciut)\not\in I(\P,\Mm)$. Then for any clause 
of $\ground{\P}$ of the form \relax

\hfil $\rR(\vec c\ciut)\klaus \rP_1(\vec c_1),\ldots,\rP_r(\vec c_r),
\overline\rS_1(\vec d_1),\ldots,\overline\rS_s(\vec d_s)$,

\noindent
either some $\rP_i(\vec c_i)$ is not in $I(\P,\Mm)$, or some 
$\rS_i(\vec d_i)$ is in~$\Mm$. A refutation is thus defined
by co-induction.
\end{proof}

Let $D$ be a refutation. For any node $\nu$ of $D$ we define 
a set of atoms $\Delta_\nu$ by induction. 
If $\nu$ is the root then $\Delta_\nu=\pusty$. Otherwise 
$\nu$ is an immediate successor (child) of a node $\mu$.
If the labels of~$\mu$ and $\nu$ are respectively 
$\rP(\vec a)$ and $\rQ(\vec b\ciut)$ then we take 
$\Delta_\nu=\Delta_\mu\cup\{\rP\rQ(\vec a\vec b\ciut)\}$.
One can say that $\Delta_\nu$ collects the history of 
transitions between atoms along the path to~$\nu$. 
Observe that, as the distance from the root increases,
this history must eventually contain a repetition.

\begin{lemma}\label{ud35u552}
Let node $\nu$ in a~refutation $D$ be labeled~$\rR(\vec c\ciut)$.
Then \mbox{$\Gamma_\Mm,\Delta_\nu,\rR?(\vec c\ciut)\vdash\circ$}.
\end{lemma}

\begin{proof}
Let $h(\nu)$ be 0 if all children of $\nu$ are labeled \relax
by overlined atoms of the form $\overline\rQ(\vec d\ciut)$
or some label occurs twice on the path to $\nu$.
Otherwise take $h(\nu)$ to be one plus the
maximum of $h(\pi)$, for all children~$\pi$ of~$\nu$.
The measure $h(\nu)$ is well-defined, because on any infinite path 
labels must repeat. The proof goes by induction on~$h(\nu)$.

First suppose that an atom $\rP(\vec c\ciut)$ occurs twice at the 
path leading to~$\nu$. Then the set $\Delta_\nu$ contains 
a sequence of assumptions of the form  
$\rP\rQ(\vec c,\vec d\ciut)$, $\rQ\rS(\vec d,\vec e\ciut)$, \ldots, 
 $\rT\rP(\vec a,\vec c\ciut)$.
By repeatedly using axioms~(\ref{aksj99}) we reduce our 
proof obligation $\Gamma_\Mm,\Delta_\nu,\rR?(\vec c\ciut)\vdash\circ$
to proving $\circ$ from an environment containing $\rP\rP(\vec c,\vec c\ciut)$
and finally we apply axiom~(\ref{aksjAA}).

In the other case 
we prove $\circ$ using an axiom of type~(\ref{aksj55}). This 
splits our proof into as many parallel ones as as there are clauses
in~$\P$ targeted~$\rR(\vec u)$, for some~$\vec u$. 
For every such clause $K(\vec x,\vec y)$
we need to prove the judgment 
\mbox{$\Gamma_\Mm,\Delta_\nu,\rR?(\vec c\ciut),\rK^0(\vec c\ciut)\vdash 
\overline\rK^0$}. Here we have two cases, depending on whether 
the vector~$\vec u$ can be instantiated to~$\vec c$ or not. 
The latter may happen when $\vec u$ contains a constant or a repetition
while $\vec c$ does not. But then an appropriate instance of 
axiom~(\ref{aksjaaa}) or~(\ref{aksjbbb}) can be applied to do the job.
Otherwise we apply axioms~(\ref{aksj66}) and this 
leads to further branching
of the proof. At the end of each of the branches we have the judgments
$\Gamma_\Mm,\Delta_\nu,\rR?(\vec c\ciut),\rK^0(\vec c\ciut),\ldots,\rK^m(\vec c,
\vec d\ciut)\vdash \overline\rK^m$, where $\vec d$ are all possible choices
of constants. So~now we have as many judgments to prove as there 
are immediate successors of~$\nu$ in~$D$: each one corresponds to 
a~clause $K(\vec c,\vec d\ciut)$ in 
${\it ground}(\P)$ targeted~$\rR(\vec c\ciut)$. 

Suppose that the successor $\pi$ of~$\nu$ 
corresponding to~$K(\vec c,\vec d\ciut)$ is labeled $\overline S(\vec e\ciut)$,
where \mbox{$S(\vec e\ciut)\in\Mm$}. 
Then $S(\vec e\ciut)\in\Gamma_\Mm$ and we can 
derive $\overline\rK^m$ using axiom~(\ref{aksj88}). 

Otherwise the subtree of $D$ rooted at~$\pi$ is a refutation 
of some~$\rP(\vec e\ciut)$. Now $h(\pi)<h(\nu)$ whence by the 
induction hypothesis we have a proof of the judgment
$\Gamma_\Mm,\Delta_\nu,
\rR\rP(\vec d,\vec e\ciut),\rP?(\vec e\ciut)\vdash \circ$. It remains to
call an instance of axiom~(\ref{aksj77}).
\end{proof}

\begin{lemma}\label{hd11143d37}
If $\P$ is incomplete for $\Mm$ then $\Gamma_\Mm\vdash \rB$.
\end{lemma}

\begin{proof} There is an atom $\rR(\vec c\ciut)\in\Mm-I(\P,\Mm)$.
By Lemma~\ref{uud352} we have a refutation \relax
for $\rR(\vec c\ciut)$. Apply Lemma~\ref{ud35u552} to the root
to obtain $\Gamma_\Mm,\rR?(\vec c\ciut)\vdash\circ$, and use 
axiom~(\ref{aksj00}) to derive~$\rB$.
\end{proof}

We now prove that
whenever $\rB$ is provable, the program must be 
incomplete.\footnote{Surprisingly, 
the natural conjecture that every proof of $\rB$ 
yields a refutation (a~correct refuter's winning strategy) 
is actually false. In a sense, the refuter can win by cheating, but
she cannot cheat against a winning strategy of the prover.} 

\paragraph{Derivations from $\overline\P\cup\overline\Mm:$}\quad 
Proofs of $\overline\P\cup\overline\Mm\vdash {\bf a}$, for 
${\bf a}\in D(\P)$, (cf.~Lemma~\ref{dsewawqd}) are 
of particularly simple shape (they do not contain any lambda-abstractions).
Such proofs can be represented as trees labeled with clauses of~$\overline\P$
and elements of $\overline\Mm$. More precisely, we say that a~finite tree $T$
is a {\it derivation\/} for $\rR(\vec e\ciut)$ when:
\begin{itemize}
\item The root of $T$ is labeled $K(\vec a,\vec b\ciut)$, 
where $K(\vec a,\vec b\ciut)$ 
is a clause in $\overline\P$ of the form

\hfill $\rR(\vec e\ciut)\klaus \rP_1(\vec c_1),\ldots,\rP_r(\vec c_r),
\overline\rS_1(\vec d_1),\ldots,\overline\rS_s(\vec d_s)$.\hfill(*)

\item The immediate subtrees of~$T$ are 
derivations of 
$\rP_1(\vec c_1),\ldots,\rP_r(\vec c_r)$ and 
single nodes labeled 
$\overline\rS_1(\vec d_1),\ldots,\overline\rS_s(\vec d_s)\in\Mm$.
\end{itemize}
\noindent
Let~$K(\vec a,\vec b\ciut)$ be of the form~(*) and assume that 
$\vec b=b_1\ldots b_d$. 
We say that the following atoms (where $i=0,\ldots,d$ and $j= 1,\ldots,r$)
are {\it justified\/} by clause~$K(\vec a,\vec b\ciut)$:
 
\hfil  $\rR?(\vec e\ciut)$, 
\hfil $\rK^i(\vec e,b_1\ldots b_i)$, 
\hfil  $\rR\rP_j(\vec e,\vec c_j)$.

\noindent
If a node $n$ of $T$ is labeled by a clause which justifies an atom~${\bf a}$
then we also say that the atom~${\bf a}$ is {\it justified\/} by the node~$n$.
We say that $T$ contains a {\it return\/} if there are two nodes in~$T$ 
such that one is a descendant of the other and both justify the same 
atom of the form $\rQ?(\vec e\ciut)$. Clearly, a~return can be ``pumped out''
so the following lemma is quite immediate.

\begin{lemma}\label{ppwe3232w}
An atom $\rQ(\vec e\ciut)$ belongs to $I(P,\Mm)$ \iff 
there exists a derivation for $\rQ(\vec e\ciut)$ without returns.
\end{lemma}

\paragraph{Justified environments:}\quad Let $T$ be a derivation.
An atom of the form $\rQ?(\vec e\ciut)$ or of the form 
$\rK^i(\vec e,\vec b\ciut)$ is {\it justified 
by~$T$\/} when it is justified by some clause 
occurring as a label in~$T$. An~atom of the form 
$\rP\rQ(\vec e,\vec c\ciut)$  is {\it justified 
by~$T$\/} when it is either justified by some clause 
occurring as a label in~$T$ or there are atoms
$\rP\rS(\vec e,\vec d\ciut)$ and $\rS\rQ(\vec d,\vec c\ciut)$ justified by~$T$.
(Informally, $T$ justifies the transitive closure of the relation 
induced by atoms justified by labels of~$T$.)

Now let $\Delta$ consist of atomic formulas of forms 
  $\rP\rQ(\vec e,\vec c\ciut)$, $\rP?(\vec e\ciut)$ and %
  $\rK^i(\vec e,\vec b\ciut)$. We say that $\Delta$ is {\it justified 
by~$T$\/} when all members of~$\Delta$ are justified by~$T$. 

\begin{lemma}\label{pwep232w3}
Let $T$ be a~derivation without returns and let $\Delta$
be justified by~$T$. Neither $\circ$ nor any of the logical constants 
$\overline\rK^i$ is provable from~$\Gamma_\Mm,\Delta$.
\end{lemma}

\begin{proof} Suppose otherwise, i.e., suppose that either $\circ$ 
or some $\overline\rK^i$ has a (long normal) proof~$M$. By induction
on the size of~$M$ we show that it is impossible. We proceed 
by cases depending on the head variable of~$M$. As a first example, consider
the case when the proof uses an assumption 

\hfil $\forall\vec z.\,\rR?(\vec z\ciut)\to  
(\rK^0_1(\vec z\ciut)\to\overline\rK^0_1)
  \to\cdots\to   (\rK^0_l(\vec z\ciut)\to\overline\rK^0_l)   \to\circ$ 

\noindent \relax
of type~(\ref{aksj55}) to derive~$\circ$. Then, for some $\vec e$, we have 
$\Gamma_\Mm,\Delta\vdash \rR?(\vec e\ciut)$ and 
$\Gamma_\Mm,\Delta,K^0_i(\vec e\ciut)\vdash \overline\rK^0_i$, for all~$i$.
There is no axiom where $\rR?$ occurs in the target, so it must be 
the case that $\rR?(\vec e\ciut)\in\Delta$. Thus $\rR?(\vec e\ciut)$ 
is justified
by some clause $K_i(\vec a,\vec b\ciut)$ occurring in~$T$. Therefore also
$\rK^0_i(\vec e\ciut)$ is justified by~$T$ (for this particular~$i$), whence 
the environment $\Delta,K^0_i(\vec e\ciut)$ is justified by~$T$. 
This way we obtain
a~contradiction from the induction hypothesis, because the proof of 
$\Gamma_\Mm,\Delta,K^0_i(\vec e\ciut)\vdash \overline\rK^0_1$
must be shorter than our proof of $\Gamma_\Mm,\Delta\vdash \circ$.

The constant $\circ$ could be proven using the 
transitivity scheme~(\ref{aksj99}):

\hfil 
$\forall\vec z\vec y\vec w (\rR\rP(\vec z,\vec y)\to \rP\rQ(\vec y,\vec w)
\to   (\rR\rQ(\vec z,\vec w)\to\circ)   \to\circ)$

\noindent
There are $\vec e$, $\vec b$, $\vec c$ such that 
$\Gamma_\Mm,\Delta\vdash\rR\rP(\vec e,\vec b\ciut)$ and 
$\Gamma_\Mm,\Delta\vdash\rP\rQ(\vec b,\vec c\ciut)$. This can only happen
when $\rR\rP(\vec e,\vec b\ciut)$ and $\rP\rQ(\vec b,\vec c\ciut)$ are 
in $\Delta$ and therefore are justified by~$T$. By definition, also
$\rR\rQ(\vec e,\vec c\ciut)$ is justified, and we apply induction to the
proof of $\Gamma_\Mm,\Delta,\rR\rQ(\vec e,\vec c\ciut)\vdash\circ$.

If $\circ$ were proven using the axiom
$\forall\vec x (\rP\rP(\vec x,\vec x)\to \circ)$ of type~(\ref{aksjAA})
then $\rP\rP(\vec e,\vec e\ciut)\in\Delta$, for some~$\vec e$.
This is a contradiction because $\rP\rP(\vec e,\vec e\ciut)$ cannot 
be justified in a derivation without returns.

Now suppose that
a proof of $\Gamma_\Mm,\Delta\vdash \overline\rK^i$ uses 
an axiom~(\ref{aksj66}):

\noindent\hfil
$\forall\vec z \vec\nu.\, \rK^i(\vec z,\vec\nu)\kto 
(\rK^{i+1}(\vec z,\vec\nu,c_1)\kto \overline \rK^{i+1})\kto
\cdots\kto(\rK^{i+1}(\vec z,\vec\nu,c_m)\kto \overline \rK^{i+1})\kto 
\overline \rK^i$,

\noindent
where $\vec\nu$ is a vector of variables of length~$i$, and 
$c_1,\ldots,c_m$ are all the constants of~$\P$. There are 
vectors $\vec e$, $\vec d$ of constants such that 
$\Gamma_\Mm,\Delta\vdash\rK^i(\vec e,\vec d\ciut)$ and 
$\Gamma_\Mm,\Delta,\rK^{i+1}(\vec e,\vec d,c_j)\vdash \overline\rK^{i+1}$,
for all~$j$. We must have $\rK^{i}(\vec e,\vec d\ciut)\in \Delta$, as there 
is no other way to prove it, whence $\rK^{i}(\vec e,\vec d\ciut)$ is 
justified by~$T$. Therefore a clause of the form $K(\vec a,\vec b\ciut)$
occurs in~$T$, for some $\vec a$ and some 
vector~$\vec b$ extending~$\vec dc_j$.
Then $\rK^{i+1}(\vec e,\vec d,c_j)$ is justified by~$T$ and we can 
apply the induction hypothesis to 
$\Gamma_\Mm,\Delta,\rK^{i+1}(\vec e,\vec d,c_j)\vdash \overline\rK^{i+1}$.
Other cases are omitted. \relax
\end{proof}

\begin{lemma}\label{pw3rr3e442}
If $\Gamma_\Mm\vdash \rB$ then $\P$ is unsound for $\Mm\ciut{:}$ there exists
an atom $\rR(\vec e\ciut)$ such that \mbox{$\rR(\vec e\ciut)\in\Mm$} but 
\mbox{$\rR(\vec e\ciut)\not\in I(P,\Mm)$}.
\end{lemma}
\begin{proof} If $\rR(\vec e\ciut)\in I(P,\Mm)$ then any derivation
for $\rR(\vec e\ciut)$ without returns justifies 
$\rR?(\vec x)$. Therefore a proof of $\Gamma_\Mm,\rR?(\vec e\ciut)\vdash\circ$
is impossible by Lemma~\ref{pwep232w3}.\relax
\end{proof}

\begin{proposition}\label{111yyy}
The formula $\varphi$ is provable \iff $\P\models_{\sms}\rOm$.
\end{proposition}

\section{Refutations into programs}
\label{sec:soups}

Now we describe a translation
in the direction opposite to that discussed in Section~\ref{tamtam},
namely from bounded-arity intuitionistic formulas to answer set programs.
The title of this section alludes to the \relax
famous slogan "proofs into programs", because our translation 
demonstrates that  ASP, as a programming paradigm,
 corresponds to logic understood from the refuter's, rather than prover's, 
point of view. 

Let $\varphi$ be a~$\Sigma_1$ formula with atomic predicates 
of arity at most~$r$. We assume that $\varphi$ is written so that
no free variable is bound and no variable is bound twice. From 
now on the variables free in~$\varphi$ 
are called~{\it constants\/}, the word {\it variable\/} being
reserved for the bound ones. The latter may of course
become free in certain subformulas of~$\varphi$. 

Let $n$ be the length of the formula~$\varphi$. The number of constants
and variables in~$\varphi$ cannot exceed~$n$ so let us assume that 
all the constants in~$\varphi$ are among $c_1,\ldots,c_n$, 
and that all bound variables are among $x_1,\ldots,x_n$. 
Observe that the number of all atoms over $c_1,\ldots,c_n,x_1,\ldots,x_n$ 
is polynomial in~$n$. 

A sequence of constants of length~$n$ is called a {\it substitution\/}. 
We write $\psi[S]$ for the result of 
applying a substitution \mbox{$S=c_{i_1}\ldots c_{i_n}$} to a~formula~$\psi$,
i.e., substituting $c_{i_j}$ for $x_j$, whenever $x_j$ is free in~$\psi$.

Let $\Gamma$ be a set of 
$\Pi_1$ formulas of the form $\psi[S]$, where $\psi$ is a subformula
of~$\varphi$, and~$S$, $T$ be substitutions. A statement of the form 
$\Gamma\nvdash{\bf a}$, where~${\bf a}$ is a closed atom, is called
a~{\it disjudgment\/}. A~triple of the form
$\<\psi,S,T\>$ is called a~{\it question\/} asked at $\Gamma\nvdash{\bf a}$
when the formula $\psi[S]$ belongs to~$\Gamma$ and  
the target of $\psi[S][T]$ coincides with ${\bf a}$. 

Then
 $\psi=\forall\vec y_1(\sigma_1\to \forall\vec y_2(\sigma_2\to 
\cdots\to\forall\vec y_k(\sigma_k\to{\bf b})\ldots))$, where 
$\sigma_i\in \Sigma_1$, and $\,{\bf b}$ is an atom.
A question $\<\psi,S,T\>$ represents a possible proof \relax
attempt (an attempt to construct a~term in long normal form) 
with a variable of type $\psi[S]$ as a head variable and
with $\vec y_1,\ldots,\vec y_k$ instantiated by~$T$. 
The variables $\vec y_1,\ldots,\vec y_k$ are called 
the {\it top variables of\/}~$\psi$. 
Every~$\sigma_i$ is a $\Sigma_1$ formula of shape
$\tau_1\to\tau_2\to\cdots\to\tau_q\to{\bf c}_i$, 
where $\tau_j\in\Pi_1$, and $\,{\bf c}$ is an atom. 
The atom ${\bf c}_i$ is then called the $i$-th {\it subgoal\/}
in~$\psi$ and the formulas
$\tau_1,\ldots,\tau_q$ are {\it $i$-th descendants\/}
of~$\psi$. Note that top variables of $\psi$
may occur free in the target atom of $\psi$, the 
$i$-th descendants and subgoals of $\psi$. More precisely,
we have $\FV({\bf b})\cup \FV(\tau_j)\cup\FV({\bf c}_i)
\subseteq\FV(\psi)\cup\vec y_1\cup\cdots\cup\vec y_i$.

An {\it $i$-th  answer\/} to the question $\<\psi,S,T\>$ as above 
is any disjudgment of the form 
$\Gamma'\nvdash{\bf c}_i[S][T]$ such that
$\Gamma, \tau_1[S][T],\ldots,\tau_q[S][T]\subseteq\Gamma'$. 
A question {\it is answered\/} in a set $\Z$ when it has 
an $i$-th answer in~$\Z$, for some~$i$. 

The intuition to be associated with an $i$-th answer is that 
the question (proof attempt) $\<\psi,S,T\>$ is challenged \relax
at the $i$-th argument: the prover is expected to fail to prove
the formula $\sigma_i[S][T]$.

Finally we define a {\it refutation soup\/} for $\varphi$, 
for brevity called {\it soup\/}, as a~set~$\Z$ of disjudgments
such that $\,\pusty\nvdash\varphi$ belongs to~$\Z$ and every 
question asked at any $(\Gamma\nvdash{\bf a})\in\Z$ is answered in~$\Z$.
It is shown in~\cite{suz16} that:
\begin{itemize}\item If~there is a soup for  $\varphi$ then $\,\nvdash\varphi$.
\item If $\,\nvdash\varphi$ then there is a soup of size
at most~$2^{n^r}$, where $r$ is the maximum arity of the predicates
in~$\varphi$. 
\end{itemize}
We now construct a program~$\P$ such that~$\P$ has a~stable model \iff
there exists a~soup for $\varphi$ of size~$2^{n^r}$.

The domain $D_\P$ of~$\P$  consists of
all (occurrences of) subformulas of $\varphi$, all the 
constants $c_1,\ldots,c_n$, and two additional constants {\bf 0} and {\bf 1}.
The size of $D_\P$ is clearly polynomial in~$n$. 
Note that we count different occurrences of the 
same subformula as different objects. 

Every disjudgment in a soup of size~$2^{n^r}$ 
can be identified by a sequence of ${\bf 0}$s and {\bf 1}s of length~$n^r$, 
called an {\it address\/}.  One can think 
of a~soup as of a set of triples of the form $\xi:\Gamma\nvdash{\bf a}$,
where $\xi$ is the address of the disjudgment~$\Gamma\nvdash{\bf a}$.

\paragraph{The vocabulary of~$\P$:} The atomic predicates occurring in $\P$ 
are written using some abbreviations. We use $\xi,\eta,\ldots$ as 
metavariables for addresses, i.e., $\xi,\eta,\ldots$ are 
sequences of length~$n^r$, 
intended to be instantiated by ${\bf 0}$s and {\bf 1}s. We also 
write $S,T,\ldots$ for ``substitutions'', i.e., sequences 
of length~$n$ to be instantiated by constants of~$\varphi$.

The program $\P$ will use the following predicates. Each of them is accompanied
below by its intuitive meaning. 

\begin{itemize}
\item $\erD_i(\tau,\psi,S,T,U)$,\quad ``the formula $\tau$ is an $i$-th
{\bf descendant} of $\psi$, and $\tau[U]$ is obtained from~$\tau[S]$ 
by applying $T$ to the top variables of~$\psi$'';
\item $\erS_{i,\bf a}(\psi,S,T)$,\quad 
``the atom ${\bf a}$ is the result of applying $T$ to the
$i$-th {\bf subgoal} in~$\psi[S]$'';
\item $\erH_{\bf a}(\psi,S,T)$,\quad 
``the atom ${\bf a}$ is the result of applying $T$ to the 
{\bf head} atom of~$\psi[S]$'';
\item $\erG_{\bf a}(\xi)$,\quad ``the atom ${\bf a}$ is the {\bf goal}
at the address~$\xi$'';
\item $\erE(\psi, S,\xi)$,\quad ``the instance $\psi[S]$ of $\psi$ 
occurs in the assumption {\bf environment} at~$\xi$'';
\item $\overline\erE(\psi, S,\xi)$,\quad ``the above does not hold'';
\item $\erQ(\psi,S,T,\xi)$,\quad ``the triple $\<\psi,S,T\>$ is 
a~{\bf question} at~$\xi$'';
\item $\erA_i(\psi,S,T,\xi,\eta)$,\quad 
``the question $\<\psi,S,T\>$ at~$\xi$ has an $i$-th {\bf answer} at~$\eta$'';
\item $\overline\erA_i(\psi,S,T,\xi,\eta)$,\quad ``the above does not hold'';
\item $\rF$,\quad an auxiliary ``false'' for the 
contradiction clauses~(\ref{clausEEF}) and~(\ref{clausNYA});
\item $\rY(\psi,S,T,\xi)$,\quad an auxiliary predicate 
for clauses~(\ref{clausNYA}--\ref{clausYA}).

\end{itemize}

\paragraph{The clauses of~$\P$:} We begin with the most obvious clauses
determined by the syntax of~$\varphi$. There is a number of facts of
the form 
\begin{enumerate}
\item\label{cladsh1} $\erD_i(\tau,\psi,S,T,U)\klaus\,$; 
\item\label{cladsh2} $\erS_{i,\bf a}(\psi,S,T)\klaus\,$;
\item\label{cladsh3} $\erH_{\bf a}(\psi,S,T)\klaus\,$,
\end{enumerate}
where $\tau,\psi,S,T,U$, are concrete formulas and substitutions, and ${\bf a}$ 
is any atom. 
For example, if $n=4$ and 
$\psi = \forall y_1(\rR(y_1,c_2)\to 
\forall y_2(\rP(y_1,c_1)\to \rR(c_1,y_2,y_3)))$
then $\P$ contains all clauses of the form 
\mbox{$\erH_{\rR(c_1,c_4,c_2)}(\psi,*,*,c_2,*,*,c_4,*,*)\klaus\,$},
where every asterisk can be replaced by any constant.
This clause can be written as $\erH_{\rR(c_1,c_4,c_2)}(\psi,S,T)\klaus\,$,
where \mbox{$S=(*,*,c_2,*)$} and \mbox{$T=(*,c_4,*,*)$}. Let 
\mbox{$\psi= \sigma\to \forall y_3((\cdots\to \tau \to\cdots\to {\bf a})
\to \forall y_1(\sigma'\to {\bf b}))$} be
another example, where $y_4$ is free and 
$\tau=\forall y_2(\rP(y_2,y_1)\to \rR(y_3,y_4))$.
Then $\P$ contains e.g., all clauses of the form 
$\erD_2(\tau,\psi,S,T,U)\klaus\,$,
for $S=(*,*,*,d\ciut)$, $T=(*,*,e,*)$, and $U=(*,*,e,d\ciut)$, 
where the asterisks
can be replaced by anything. This is because variables $y_1$ and~$y_2$
are bound in both~$\psi$ and~$\tau$, and therefore ignored by the 
substitutions. There is a lot of similar clauses 
but still only a~polynomial number of such clauses.
The remaining clauses determine the shape of the model which is 
supposed to represent a~soup. We begin with a few facts 
describing the initial judgment at address ${\bf 00}\ldots{\bf 0}$.
If $\varphi=\psi_1\to\cdots\to\psi_n\to{\bf a}$ then the following
clauses are in~$\P$:

\begin{enumerate}\setcounter{enumi}{3}
\item\label{clausG}$\erG_{\bf a}({\bf 00}\ldots{\bf 0})\klaus\,$;
\item\label{clausE} $\erE(\psi_i,S,{\bf 00}\ldots{\bf 0})\klaus\,$;
\item\label{clausEF} $\overline\erE(\psi,S,{\bf 00}\ldots{\bf 0})\klaus\,$,
\end{enumerate}
where $i=1,\ldots,n$, $\psi\not\in\{\psi_1,\ldots,\psi_n\}$, 
and $S$ is an arbitrary substitution (note
that $\psi_i$ have no free variables, so $S$ does not matter).

The following clauses guarantee that any answer 
given in the model is correct \wrt the question it responds to.

\begin{enumerate}\setcounter{enumi}{6}
\item\label{clausEAE}
$\erE(\tau,U,\eta)\klaus\erA_i(\psi,S,T,\xi,\eta),\,\erE(\tau,U,\xi)$;
\item \label{clausEAD}
$\erE(\tau,U,\eta)\klaus\erA_i(\psi,S,T,\xi,\eta),\,
\erD_i(\tau,\psi,S,T,U)$;

\item\label{clausGA} $\erG_{\bf a}(\eta) \klaus\erA_i(\psi,S,T,\xi,\eta),\,
\erS_{i,\bf a}(\psi,S,T)$. 
\end{enumerate}

\noindent
The last one of the three above clauses defines the goal at~$\eta$,
while the first two list the necessary assumptions at~$\eta$. 
But the definition of a soup allows any other formula to occur as
an assumption in the answer judgment. This must be reflected by
any stable model. For this purpose we add to~$\P$ the three clauses 
(using the special symbol~$\rF$):
\begin{enumerate}\setcounter{enumi}{9}
\item\label{clausEE} $\erE(\psi,S,\xi)\klaus \neg\overline\erE(\psi,S,\xi)$,
\quad $\overline\erE(\psi,S,\xi)\klaus \neg\erE(\psi,S,\xi)$,
\item\label{clausEEF} 
$\rF \klaus \erE(\psi,S,\xi),\overline\erE(\psi,S,\xi),\neg\rF$,
\end{enumerate}
which force that either $\erE(\psi,S,\xi)$ or $\overline\erE(\psi,S,\xi)$
must hold, but not both. 

A similar measure is also applied for the predicates~$\erA_i$
(so far only occurring at the rhs). 
That is, $\P$
contains (for $i=1,\ldots n$) the constructs: 

\begin{enumerate}\setcounter{enumi}{11}
\item\label{clausAA} $\erA_i(\psi,S,T,\xi,\eta)\klaus\neg\overline
\erA_i(\psi,S,T,\xi,\eta),\erQ(\psi,S,T,\xi)$,\\ \quad 
$\overline\erA_i(\psi,S,T,\xi,\eta)\klaus 
\neg\erA_i(\psi,S,T,\xi,\eta),\erQ(\psi,S,T,\xi)$.
\end{enumerate}

\noindent
The next clauses (one for each ${\bf a}$)
define the notion of a question.

\begin{enumerate}\setcounter{enumi}{12}
\item \label{clausQEHG}
$\erQ(\psi,S,T,\xi)\klaus \erE(\psi, S,\xi),\, \erH_{\bf a}(\psi,S,T), \,
\erG_{\bf a}(\xi)$.
\end{enumerate}

\noindent
Clauses~\ref{clausAA} permit an arbitrary choice 
between predicates $\erA_i(\dots)$ and $\overline\erA_i(\dots)$.
We need additional clauses to ensure that this choice makes sense, 
i.e.,  that every question has an answer. For this purpose $\P$ 
includes clauses (where $i\leq n$, and $\rY$ does not occur elsewhere).

\begin{enumerate}\setcounter{enumi}{13}
\item\label{clausNYA}
 $\rF\klaus \neg \rY(\psi,S,T,\xi),\
\erQ(\psi,S,T,\xi),\neg\rF$;
\item\label{clausYA} $\rY(\psi,S,T,\xi)\klaus \erA_i(\psi,S,T,\xi,\eta)$.
\end{enumerate}
\noindent
Clauses~(\ref{clausNYA}--\ref{clausYA}) together enforce that 
any stable model of the program must satisfy the following
property: If $\erQ(\psi,S,T,\xi)\in\Mm$ then 
\mbox{$\erA_i(\psi,S,T,\xi,\eta)\in\Mm$}, for some~$i$ and~$\eta$. 
Indeed, in a~stable model it is impossible to have 
$\erQ(\psi,S,T,\xi)$ without $\rY(\psi,S,T,\xi)$. And $\rY(\psi,S,T,\xi)$
can only be derived if $\erA_i(\psi,S,T,\xi,\eta)$ holds 
for some~$i$ and~$\eta$.\footnote{This is an example of a more general
pattern: clauses $\rY(\vec y)\klaus \rP(\vec x,\vec y\ciut)$ 
and~$\klaus\neg\rA(\vec y\ciut)$, where $\rY$ is fresh, enforce that
any stable model must satisfy 
$\forall\vec y\,\exists\vec x\,P(\vec x,\vec y\ciut)$.}

At the end we need to guarantee that no judgment can address 
two goals. This is handled by one clause for any pair ${\bf a}$ and ${\bf b}$
of distinct atoms.
\begin{enumerate}\setcounter{enumi}{15}
\item\label{clausGG} $~\klaus \erG_{\bf a}(\xi),\,\erG_{\bf b}(\xi)$.
\end{enumerate}

\paragraph{Cooking a soup from a model} We now show that every
stable model~$\Mm$ of $\P$ defines a~soup for~$\varphi$. 
Clauses~(\ref{clausEE}) and~(\ref{clausEEF}) force $\Mm$ to assign 
an environment $\Gamma_\xi$ to every address~$\xi$. This~$\Gamma_\xi$
is uniquely defined. Indeed, for each formula $\psi[S]$
either $\erE(\psi,S,\xi)\in\Mm$ or $\overline\erE(\psi,S,\xi)\in\Mm$
by clause~(\ref{clausEE}), and not both of them by clause~(\ref{clausEEF}).
(Note that~(\ref{clausEE}) alone are not enough, because $\erE$ and
$\overline\erE$ may occur at the lhs of other clauses.)

\begin{lemma}\label{hhderoo22}
Let $\Z$ consists of all disjudgments
$\xi:\Gamma_\xi\nvdash{\bf a}$ such that 
\mbox{$\erG_{\bf a}(\xi)\in\Mm$}. Then $\Z$ is a~soup. 
\end{lemma}

\begin{proof}
Clauses~(\ref{clausG}--\ref{clausEF}) ensure that
the initial judgment is in~$\Z$. 
In addition, whenever
\mbox{$\erG_{\bf a}(\xi)\in\Mm$}, and the triple $\<\psi,S,T\>$ is a question
at~$\xi$, then this question has an answer in~$\Z$.
Indeed, clause~(\ref{clausQEHG}) yields $\erQ(\psi,S,T,\xi)\in\Mm$, \relax
and this implies that either \mbox{$\erA_i(\psi,S,T,\xi,\eta)\in\Mm$} or
$\overline\erA_i(\psi,S,T,\xi,\eta)\in\Mm$ by clause~(\ref{clausAA}). 
Because of~(\ref{clausNYA})
we must have $\rY(\psi,S,T,\xi)\in\Mm$ 
and since $\Mm$ is stable, one of clauses~(\ref{clausYA}) must have been fired,
i.e., at least one~$\erA_i(\psi,S,T,\xi,\eta)$ is in the model.
It remains to observe the effect of clauses~(\ref{clausEAE}--\ref{clausGA}):
the disjudgment at~$\eta$ is indeed an answer to the question $\<\psi,S,T\>$
asked at~$\xi$.
\end{proof}

\paragraph{Boiling out a model from a soup} Suppose $\Z$ is a soup without
unnecessary ingredients, that is
every judgment in $\Z$, except the initial judgment, is an answer
to some question posed in~$\Z$. Assume also that every judgment
in the soup is given an address, with the initial judgment
having the address ${\bf 00}\ldots{\bf 0}$. We require that 
any given address uniquely determines a judgment, but one judgment
may have several addresses. Repetitions are permitted, because 
it is technically convenient that all addresses of a fixed length 
are in use (do actually refer to some judgments).

We say that 
the address ${\bf 00}\ldots{\bf 0}$ has {\it depth\/}~zero.
Otherwise the depth of an address is 1 plus the minimal depth
of a question it answers. This definition is extended to closed atoms 
as follows: the {\it depth\/} of a closed atom, where a single address 
occurs (like e.g.,~$\erE(\psi,S,\xi)$) is the depth of that address, and
in case of $\erA_i(\psi,S,T,\xi,\eta)$ it is the depth of the first 
address~$\xi$.  

We define a model $\Mm$ (recall that
a model is just a set of ground atoms). The elements of $\Mm$ are 
chosen as follows:

\begin{itemize}
\item Atoms $\erD_i(\tau,\psi,S,T,U)$, $\erS_{i,\bf a}(\psi,S,T)$,  
$\erH_{\bf a}(\psi,S,T)$ are selected according to the syntax of
the main formula~$\varphi$.
\item Atoms $\erG_{\bf a}(\xi)$, $\erE(\psi, S,\xi)$, 
$\overline\erE(\psi, S,\xi)$
are selected according to the shape of judgments.
\item Atoms $\erQ(\psi,S,T,\xi)$, $\erA_i(\psi,S,T,\xi,\eta)$, and 
$\overline\erA_i(\psi,S,T,\xi,\eta)$ are selected according to the 
structure of questions and answers in the soup. But 
if $\<\psi, S,T\>$ does not form a question at address $\xi$ then
neither $\erA_i(\psi,S,T,\xi,\eta)$ nor $\overline\erA_i(\psi,S,T,\xi,\eta)$
is in the model.
\item Atom $\rF$ is not selected.
\item The atom $\rY(\psi,S,T,\xi)$ is selected whenever
at least one of $\erA_i(\psi,S,T,\xi,\eta)$ is selected.
\end{itemize}

\begin{lemma}\label{oodd4436}
The model $\Mm$ is a stable model of~$\P$.
\end{lemma}
\begin{proof}
One must check that all the selected atoms 
are derivable, i.e., belong to~$I(\P,\Mm)$. 
For the syntax-related atoms it is obvious by the choice 
of facts (\ref{cladsh1}--\ref{cladsh3}). In all other atoms there are
occurrences of addresses, and one proceeds by induction \wrt depth.
Suppose all atoms in~$\Mm$ of depth smaller than~$d$ are derivable, \relax
and let the depth of~$\xi$ be~$d$. Consider the goal~${\bf a}$ at
the address~$\xi$. If $d=0$ then $\erG_{\bf a}({\bf 00}\ldots{\bf 0})$
is derived from clause~(\ref{clausG}). Otherwise $\xi$ is an address
of an answer to a question of depth at most $d-1$. Thus we can instantiate
clause~(\ref{clausGA}) as 
$\erG_{\bf a}(\xi) \klaus\erA_i(\psi,S,T,\zeta,\xi),\,
\erS_{i,\bf a}(\psi,S,T)$, where $\zeta$ is of depth at most $d-1$,
and by the induction hypothesis $\erA_i(\psi,S,T,\zeta,\xi)\in I(\P,\Mm)$.
This derives $\erG_{\bf a}(\xi)$ which is then used to derive
atoms of the form~$\erQ(\psi,S,T,\xi)$ from~(\ref{clausQEHG}).\footnote{
This is the only way to derive $\erQ(\psi,S,T,\xi)$, so it can only happen
when the necessary atoms are in~$\Mm$, i.e., the question in question
does actually occur in the soup.}
Given $\erQ(\psi,S,T,\xi)$, clause~(\ref{clausAA}) yields either
$\erA_i(\psi,S,T,\xi,\eta)$ or $\overline\erA_i(\psi,S,T,\xi,\eta)$,
for any~$\eta$, 
depending on whether an appropriate answer occurs at $\eta$ or not. 

The proof that no unwanted atom is derivable goes by 
induction \wrt derivations in~$\P^\Mm$. First observe that 
clauses~(\ref{cladsh1}--\ref{clausEF})
only introduce atoms that belong to the model~$\Mm$. Then note 
that in each pair of clauses~(\ref{clausEE}) and~(\ref{clausAA}) 
one of the negated atoms is already in the model. As the 
corresponding clause gets erased, it is impossible to use~(\ref{clausEE}) 
or~(\ref{clausAA}) to derive both 
the underlined and non-underlined version of the same atom.
In particular, clause~(\ref{clausEEF}) will not derive~$\rF$.
Suppose an atom $\erE(\tau,U,\eta)$ is derived by clause~(\ref{clausEAE}).
Then both $\erA_i(\psi,S,T,\xi,\eta)$ and $\erE(\tau,U,\xi)$ must have
been derived before, and therefore belong to~$\Mm$. It follows that
$\eta$ provides an $i$-th answer to $\xi$ in our soup, whence 
$\erE(\tau,U,\eta)\in\Mm$ as well. Similar argument applies to 
atom derived using clauses~(\ref{clausEAD},\ref{clausGA},\ref{clausQEHG},%
\ref{clausYA}).
\end{proof}

We can now conclude with the following

\begin{proposition}\label{222xxx}
The program $\P$ has a stable model \iff $\varphi$ is refutable.
\end{proposition}

\begin{corollary}
  SMS entailment and 
monadic $\Sigma_1$ provability formulas are mutually polynomial 
time translatable.
\end{corollary}

\section*{Conclusion} 

In this paper we demonstrated a mutual polynomial time translation
between first-order Answer Set Programming and the bounded arity 
$\Sigma_1$ fragment of intuitionistic predicate logic. One of the main 
motivations for this work was to give an alternative 
point of view on ASP, namely the constructive interpretation. 
Constructive proofs serve to derive conclusions from assumptions, 
very much as logic programs derive goals, therefore we find it natural 
to work in the context of the ASP entailment problem. Especially because
existence of stable models is nothing else than non-entailment. 

While our translations may seem complicated on the first look,
let us point out that they are actually quite natural. 
The format of the object (target) language is very simple, in particular 
very similar to the ASP source language. All the assumptions in the resulting
formul have the form of straightforward rules where each premise is
either an atom or a pseudo-negated atom of
the form~\mbox{${\bf a}\to\star$}. With such rules as axioms, 
proof search can be seen
as an algorithm which manipulates a database by expanding and querying it.
The translation therefore not only yields the known complexity bound but 
also delineates the spectrum of adequate proof tactics.

Our work gives a complementary sight on the explanation of logic
programming in terms of Howard's system \textbf{H} introduced by Fu
and Komendantskaya \cite{fukom17}. The system interprets Horn formulas
as types, and derivations for a given formula as the proof terms
inhabiting the type corresponding to the formula. The authors
demonstrate that different forms of resolution can be expressed in
this system and relate them. In our work the operation of resolution
is limited since we work with constants only. Instead, our work
concentrates on the operational semantics of negation. Indeed, our
translation, due to the limited number of available formula targets,
has strong flavour of countinuation-passing style transformation.

Logic programming with negation undestood as in ASP has its
explanation in terms of equilibrium logic
\cite{Pearce99,Pearce06}. This logic is built on the foundations of
intuitionistic theory called \emph{here and there logic}, which can be
axiomatized in propositional case by the scheme
\begin{displaymath}
(\lnot\alpha\to\beta)\to \varphi_P(\alpha,\beta), 
\end{displaymath}
\noindent 
where $\varphi_P(\alpha,\beta) = ((\beta\to\alpha)\to\beta)\to\beta$
is the Peirce's law. It is interesting to see that proof construction
techniques that occur when the result of our translation is to be
proved are very similar to ones used in case of \emph{here and there
  logic}.

While the equilibrium logic relies on a property of a model to settle
the space to represent a~coherent database of inferrable basic
statements, in the case of our logic, the space is located in the
proof environment that evolves during the proof construction process.

\subsection*{Acknowledgement} The authors are greatly indebted to Konrad 
Zdanowski for numerous discussions and suggestions regarding this work.

\bibliographystyle{acmtrans}
\bibliography{asp}

\begin{thebibliography}{}

\bibitem[\protect\citeauthoryear{Brewka, Eiter, and Truszczyński}{Brewka
  et~al\mbox{.}}{2011}]{Brewkaetal}
{\sc Brewka, G.}, {\sc Eiter, T.}, {\sc and} {\sc Truszczyński, M.} 2011.
\newblock Answer set programming at a glance.
\newblock {\em Commun. ACM\/}~{\em 54,\/}~12 (Dec.), 92--103.

\bibitem[\protect\citeauthoryear{Dantsin, Eiter, Gottlob, and Voronkov}{Dantsin
  et~al\mbox{.}}{2001}]{dantsin}
{\sc Dantsin, E.}, {\sc Eiter, T.}, {\sc Gottlob, G.}, {\sc and} {\sc Voronkov,
  A.} 2001.
\newblock Complexity and expressive power of logic programming.
\newblock {\em {ACM} Comput. Surv.\/}~{\em 33,\/}~3, 374--425.

\bibitem[\protect\citeauthoryear{Fu and Komendantskaya}{Fu and
  Komendantskaya}{2017}]{fukom17}
{\sc Fu, P.} {\sc and} {\sc Komendantskaya, E.} 2017.
\newblock Operational semantics of resolution and productivity in {H}orn clause
  logic.
\newblock {\em Formal Asp. Comput.\/}~{\em 29,\/}~3, 453--474.

\bibitem[\protect\citeauthoryear{Kolaitis and Papadimitriou}{Kolaitis and
  Papadimitriou}{1988}]{kolpap}
{\sc Kolaitis, P.~G.} {\sc and} {\sc Papadimitriou, C.~H.} 1988.
\newblock Why not negation by fixpoint?
\newblock In {\em Proceedings of the Seventh ACM SIGACT-SIGMOD-SIGART Symposium
  on Principles of Database Systems}. {\it PODS'88}. ACM, New York, NY, USA,
  231--239.

\bibitem[\protect\citeauthoryear{Lifschitz}{Lifschitz}{2008}]{Lifschitz}
{\sc Lifschitz, V.} 2008.
\newblock What is answer set programming?
\newblock In {\em Proceedings of the 23rd National Conference on Artificial
  Intelligence - Volume 3}. AAAI'08. AAAI Press, 1594--1597.

\bibitem[\protect\citeauthoryear{Marek and Truszczy\'{n}ski}{Marek and
  Truszczy\'{n}ski}{1991}]{martru}
{\sc Marek, W.} {\sc and} {\sc Truszczy\'{n}ski, M.} 1991.
\newblock Autoepistemic logic.
\newblock {\em J. ACM\/}~{\em 38,\/}~3 (July), 587--618.

\bibitem[\protect\citeauthoryear{Osorio, Navarro, and Arrazola}{Osorio
  et~al\mbox{.}}{2004}]{Osorio04}
{\sc Osorio, M.}, {\sc Navarro, J.~A.}, {\sc and} {\sc Arrazola, J.} 2004.
\newblock Applications of intuitionistic logic in answer set programming.
\newblock {\em Theory Pract. Log. Program.\/}~{\em 4,\/}~3 (May), 325--354.

\bibitem[\protect\citeauthoryear{Pearce}{Pearce}{1999}]{Pearce99}
{\sc Pearce, D.} 1999.
\newblock Stable inference as intuitionistic validity.
\newblock {\em The Journal of Logic Programming\/}~{\em 38,\/}~1, 79--91.

\bibitem[\protect\citeauthoryear{Pearce}{Pearce}{2006}]{Pearce06}
{\sc Pearce, D.} 2006.
\newblock Equilibrium logic.
\newblock {\em Annals of Mathematics and Artificial Intelligence\/}~{\em
  47,\/}~\mbox{1-2}, 3--41.

\bibitem[\protect\citeauthoryear{Schubert and Urzyczyn}{Schubert and
  Urzyczyn}{2018}]{su-asp}
{\sc Schubert, A.} {\sc and} {\sc Urzyczyn, P.} 2018.
\newblock Answer set programming in intuitionistic logic.
\newblock {\em Indagationes Mathematicae\/}~{\em 29,\/}~1, 276--292.
\newblock L.E.J. Brouwer, fifty years later.

\bibitem[\protect\citeauthoryear{Schubert, Urzyczyn, and
  Walukiewicz-Chrząszcz}{Schubert et~al\mbox{.}}{2015}]{suw2013-www}
{\sc Schubert, A.}, {\sc Urzyczyn, P.}, {\sc and} {\sc Walukiewicz-Chrząszcz,
  D.} 2015.
\newblock Restricted positive quantification is not elementary.
\newblock In {\em Proc.~TYPES 2014}, {H.~Herbelin}, {P.~Letouzey}, {and}
  {M.~Sozeau}, Eds. LIPIcs, vol.~39. Schloss Dagstuhl--Leibniz-Zentrum f\"ur
  Informatik, 251--273.

\bibitem[\protect\citeauthoryear{Schubert, Urzyczyn, and Zdanowski}{Schubert
  et~al\mbox{.}}{2016}]{suz16}
{\sc Schubert, A.}, {\sc Urzyczyn, P.}, {\sc and} {\sc Zdanowski, K.} 2016.
\newblock On the {M}ints hierarchy in first-order intuitionistic logic.
\newblock {\em Logical Methods in Computer Science\/}~{\em 12,\/}~4.

\end{thebibliography}

\label{lastpage}
\end{document}